\documentclass[leqno,11pt]{amsart}
\usepackage{amsmath,amssymb,amsfonts,amsthm}
\usepackage{amsmath,amssymb,amsfonts,amsthm}
\usepackage{graphics}
\usepackage{graphicx,color}

\newtheorem{theorem}{Theorem}

\newtheorem{lemma}[theorem]{Lemma}

\newtheorem{remark}[theorem]{Remark}

\begin{document}
\title[ Pricing of Basket Options]{Pricing of Basket Options Using Polynomial Approximations }
\author[Olivares, Pablo]{Pablo Olivares}
\address{Department of Mathematics, Ryerson University}

\begin{abstract}
In this paper we use Bernstein and Chebyshev polynomials to approximate the price of some basket options under a bivariate Black-Scholes model. The method consists in expanding the price of a univariate related contract after conditioning on the remaining underlying assets  and  calculating  the mixed exponential-power moments of a Gaussian distribution that arise as a consequence of such approximation. Our numerical implementation  on spreads contracts shows the method is as accurate as a standard Monte Carlo approach at considerable lesser computational effort.
\end{abstract}

\maketitle

\section{Introduction}
The objective of the paper is the study of  basket options under  a standard multivariate Black-Scholes model using different polynomial approximations  of a related conditional contract on one of the underlying assets. \\
By conditioning on some of the asset prices we reduce the problem to computing the expected value of the Black-Scholes formula with a random strike price. Under suitable polynomial approximations this expected value can be obtained from  the corresponding one dimensional price at selected points together with the  truncated exponential-power moments of a joint Gaussian distribution. It allows to compute prices with a precision comparable to a Monte Carlo approach and considerable less computational effort. \\
An alternative point of view using Taylor polynomials is  followed in  Li, Deng and Zhou (2008, 2010) or Alvarez and Olivares (2014).  Although a Taylor approximation produces fairly estimates in terms of a simple closed-formula based on the derivatives of the conditional price, it is quite sensible to the point around the expansion is done. Moreover, as this expansion is  local at a particular value of the parametric set, it may introduce significant errors, albeit infrequent, at values far from the point where the expansion is considered.\\
In order to overcome this potential problem we study  developments in terms of Bernstein and Chebyshev polynomial approximations, which offer a uniform convergence of the conditional price on a given closed interval. For the seek of simplicity we focus on bidimensional spread contracts and a Black-Scholes model, but the method can be extended to more complex models or other European derivatives. See for example Olivares (2014) for applications to the case of jump-diffusion and switching jump-diffusion models. Interesting approximations in terms of Fourier series can be found in Meng and Ding(2013) and in Fang and Oosterlee(2009).\\
Basket options are a multivariate extensions of univariate European calls or puts. A basket option takes the weighted average of a group of $d$ stocks (the basket) as the underlying, and produces a payoff equal to the maximum of zero and the difference between the weighted average and the strike (or the opposite difference for the case of a put). Index options, whose value depend on the movement of an equity or other financial index such as the S\&P 500 and real options based on the difference between gas and oil prices are examples of such contracts.\\
 In the particular case of spread options, several approximations have been previously considered in the works of Kirk(1995), Carmona and Durrleman(2003), Li, Deng  and Zhou(2008, 2010) and Venkatramanan and Alexander (2011),  where different ad-hoc approaches are studied.  \\
 %Moreover, Fast Fourier Transform methods have been successfully  used to compute spread prices under some Levy processes, see  Hurd and Zhou (2009) %and Cane and Olivares (2013).\\
 The organization of the paper is the following, in section 2 we introduce the model, main notations and derive Bernstein  approximation for spread options. In section 3 we study the case of a Chebyshev approximation and the sensitivity with respect to the spot prices. In section 4 we discuss the  numerical implementation and results. Finally in section 5 we conclude.
\section{An Approximation Based on Bernstein Polynomials}
We first introduce some notations and our main model.\\
 Let  $(\Omega ,\mathcal{A}, \mathbb{P},\{\mathcal{F}_{t}\}_{t>0})$ be a filtered probability space and define the filtration  $\mathcal{F}^{X_t}:= \sigma(X_s, 0 \leq s \leq t)$ as the $\sigma$-algebra generated by the random variables $\{X_s, 0 \leq s \leq t \}$ completed in the usual way. Denote  by $\mathcal{Q}$ the risk neutral measure and $E_{\mathcal{Q}}$ the expectation under $\mathcal{Q}$.\\
 The quantities $\mu_{a,b}$ and $M_X(u,a,b)$ represent respectively the truncated moment and the generating moment function (g.m.f.) on $[a,b]$, while the function $N(.)$ is the cumulative distribution function (c.d.f.) of a standard normal distribution.\\
 The matrix $A'$ represents the transpose of  matrix $A=(a_{ij})$, while $diag(A)$ is, depending on the context, a diagonal matrix with components $a_{ii}$ in the main diagonal or a column vector with components from the diagonal of the matrix $A$.  On the other hand $A^{\frac{1}{2}}$ denotes a matrix such that $A^{\frac{1}{2}} A^{\frac{1}{2}}=A$.  The d-dimensional column vector of ones is described by $1_d$.\\
The process of spot prices is denoted by $S_t=\left(S^{(1)}_t,S^{(2)}_t,\ldots,S^{(d)}_t \right)'_{0 \leq t \leq T}$ while $Y_t=(Y_t^{(1)},Y_t^{(2)},\ldots,Y_t^{(d)})'_{0 \leq t \leq T}$ defines the asset log-return process at time $t$, related as
\begin{equation}\label{}
    S^{(j)}_t= S^{(j)}_0 \exp(Y_t^{(j)})\;\; \text{for}\;\; j=1,2,\ldots,d
\end{equation}
We analyze European basket options whose payoff, at maturity time $T$ and  strike price $K$, is given by:
\begin{equation*}\label{eq:basket}
    h(S_T)=\left( \sum_{j=1}^d w_j S_T^{(j)}-K \right)_+
\end{equation*}
where $(w_j)_{1 \leq j \leq d}$ are some deterministic weights and $x_+=max(x,0)$. \\
Probably the most studied of these contracts  are spread options whose payoff is
\begin{equation*}\label{eq:spread}
    h(S_T)=(  S_T^{(1)}- S_T^{(2)}-K )_+
\end{equation*}
We will focus on the case $d=2$  assuming a bidimensional Black-Scholes dynamics, that  under the risk neutral probability is given by:
\begin{equation}\label{eq:bscmultid}
    dS_t=rS_t dt+ \Sigma^{\frac{1}{2}} S_t dB_t
\end{equation}
where $B_t=(B^{(1)}_t, B^{(2)}_t)_{t \geq 0}$ is a two dimensional vector of  Brownian motions such that $<B^{(1)}_t,B^{(2)}_t>=\rho t$ and $\Sigma$ is a positive definite symmetric matrix defined by:
\begin{equation*}\label{}
    \Sigma=\left(\begin{array}{ll}
                   \sigma_1^2 & 0 \\
                 0  & \sigma_2^2
                 \end{array}
                 \right)
\end{equation*}
 Here $r>0$ is the (constant) interest rate. \\
Equivalently, after applying Ito formula:
\begin{equation}\label{eq:retbscmultid}
    dY_t=(r-\frac{1}{2} diag(\Sigma)) dt+ \Sigma^{\frac{1}{2}}  dB_t
\end{equation}
  The price of a basket contact with maturity at $T>0$ and payoff $h(S_T)$, denoted by $C_S:=C_S(S_0,K,\Sigma,\rho,r,T))$, is given by:
\begin{eqnarray}
C_S &=&e^{-rT}E_{\mathcal{Q}}h(S_{T})=E_{\mathcal{Q}}\left[ e^{-rT}E_{\mathcal{Q}}\left(h(S_{T})|\mathcal{F}^{Y^{(2)}_T} \right)\right]
\notag \\
&=&E_{\mathcal{Q}}[C (Y^{(2)}_T)]
 \label{eq:gen-price}
\end{eqnarray}
where:
\begin{equation}\label{eq:condprice3}
C (y):=e^{-rT} E_{\mathcal{Q}}\left( h(S_{T})|\mathcal{F}^{Y^{(2)}_T} \right)\mid_{Y^{(2)}_T=y}
\end{equation}
represents the price of a one dimensional European contract with underlying  $S_T^{(1)}$ and payoff $h(S_T)$ when $Y_T^{(2)}=y$.\\
We start  writing the price $C_S$ in terms of a conditional  price via the following  elementary lemma.
\begin{lemma} \label{lemma1}
The price  of a European basket contract with payoff $h(S_T)$  under the model given by equation (\ref{eq:bscmultid}) is
\begin{equation}\label{eq:spreadcond}
    C_S= w_1 e^A E_{\mathcal{Q}} \left(e^{\frac{\sigma_1}{\sigma_2}\rho Y^{(2)}_T}C (Y^{(2)}_T) \right)
\end{equation}
where:
\begin{equation*}\label{eq:bschkvariab}
    C(y):=C_{BS}(K(y),\sigma,S_0^{(1)})
\end{equation*}
 is the Black-Scholes price of a Call Option written on the underlying asset price $S_T^{(1)}$ with strike price $K(y)$, maturity at $T>0$, volatility  $\sigma$, spot price $S_0^{(1)}$ and strike price given by:
  \begin{equation*}\label{}
    K(y)=\frac{1}{w_1} e^{-A} \left( Ke^{-\frac{\sigma_1}{\sigma_2}\rho y}- w_2 S_0^{(2)}e^{(1-\frac{\sigma_1}{\sigma_2}\rho)y} \right)
  \end{equation*}
  with:
\begin{eqnarray*}\label{constanta}
    A &=& - \frac{1}{2}  \frac{\sigma_1}{\sigma_2} \rho (\rho \sigma_1 \sigma_2 +2r -\sigma_2^2)T\\ \notag
    \sigma &= & \sqrt{1-\rho^2} \sigma_1
   \end{eqnarray*}
\end{lemma}
\begin{proof}
From equation (\ref{eq:retbscmultid}) we have:
 \begin{equation*}  \label{eq:risk-neutral-constant}
Y_T=(Y_T^{(1)},Y_T^{(2)}) \sim N \left( (r1_2-\frac{1}{2}  diag(\Sigma))T,
\Sigma_{\rho} T \right)
\end{equation*}
where:
\begin{equation*}\label{}
    \Sigma_{\rho}=\left(\begin{array}{ll}
                   \sigma_1^2 & \rho \sigma_1 \sigma_2 \\
                 \rho \sigma_1 \sigma_2   & \sigma_2^2
                 \end{array}
                 \right)
\end{equation*}
Moreover, it is well known that the conditional  distribution of $Y_T^{(1)}$ given $Y_T^{(2)}$ is also normal. See for example Tong(1989). Thus we can write
\begin{equation}\label{eq:condnormalrv}
Y_T^{(1)}= \mu(Y^{(2)}_T)+\sigma \sqrt{T}Z\;\;\;\text{in law}
\end{equation}
where:
\begin{equation*}\label{eq:mu}
 \mu( Y^{(2)}_T)=(r(1-\frac{\sigma_1}{\sigma_2} \rho)+ \frac{1}{2}\sigma_1 \sigma_2 \rho-\frac{1}{2} \sigma_1^2) T +\frac{\sigma_1}{\sigma_2} \rho Y_T^{(2)}
\end{equation*}
and  $Z \sim N(0,1)$, independent of $Y_T^{(2)}$.\\
  Next, from equation (\ref{eq:gen-price}) we have:
\begin{eqnarray} \notag
C_S &=& w_1 e^{-rT}E_{\mathcal{Q}} \left(E_{\mathcal{Q}} \left[\left( S_0^{(1)}e^{Y_T^{(1)}}+ \frac{w_2}{w_1}S_0^{(2)}e^{Y_T^{(2)}} -\frac{K}{w_1} \right)_+\mid \mathcal{F}^{Y^{(2)}_T} \right] \right) \\ \notag
&=& w_1 e^{-rT}E_{\mathcal{Q}} \left(E_{\mathcal{Q}} \left[\left( S_0^{(1)}e^{Y_T^{(1)}}-K'(Y^{(2)}_T)  \right)_+\mid \mathcal{F}^{Y^{(2)}_T} \right] \right) \\ \label{eq:pgral}
 \end{eqnarray}
 where $K'(y)=\frac{K}{w_1}- \frac{w_2}{w_1}S_0^{(2)}e^{y}$.\\
 Substituting equation (\ref{eq:condnormalrv}) into (\ref{eq:pgral}) we have:
 \small{
 \begin{eqnarray}\notag
    C_S   &=& w_1 e^{-rT}E_{\mathcal{Q}}\left[e^{\mu({Y_T^{(2)})}} E_{\mathcal{Q}} \left( \left( S_0^{(1)}e^{\sigma\sqrt{T}Z}-e^{-\mu({Y_T^{(2)})}}K'(Y^{(2)}_T)  \right)_+ \mid \mathcal{F}^{Y^{(2)}_T} \right) \right]\\ \notag
   &=& w_1 e^{-rT}\\ \notag
 && E_{\mathcal{Q}}\left[e^{-(r-\frac{1}{2}\sigma^2) T+\mu({Y_T^{(2)})}} E_{\mathcal{Q}}\left( \left( S_0^{(1)}e^{(r-\frac{1}{2}\sigma^2)T+\sigma\sqrt{T}Z}-K(Y^{(2)}_T) \right)_+ \mid \mathcal{F}^{Y^{(2)}_T} \right)\right] \\ \notag
 &=&  w_1  E_{\mathcal{Q}}\left[e^{-(r-\frac{1}{2}\sigma^2) T+\mu({Y_T^{(2)})}}C(Y^{(2)}_T) \right] \\ \label{eq:bsspiecebs2}
 &&
   \end{eqnarray}
  }
where:
  \begin{equation*}
  C(Y^{(2)}_T):= C_{BS}(K(Y^{(2)}_T),\sigma,S_0^{(1)})= e^{-rT} E_{\mathcal{Q}}\left[\left( S_0^{(1)}e^{(r-\frac{1}{2}\sigma^2)T+\sigma\sqrt{T}Z}
 -  K(Y^{(2)}_T) \right)_+ \mid \mathcal{F}^{Y^{(2)}_T}\right]
  \end{equation*}
  Equation (\ref{eq:spreadcond}) follows immediately.
\end{proof}
 For any $n \in \mathbb{N}$ and $[a, b]  \subset \mathbb{R}$  we consider the n-th expansion of the function $C(y)$ in terms of Bernstein polynomials  on the interval $[a,b]$. Denoted as  $\tilde{C}_B(y,a,b,n)$ it is defined by:
\begin{eqnarray}\label{eq:bernsexp}
 \tilde{C}_B(y,a,b,n)  &=& \sum_{\nu=0}^n C((b-a)\frac{\nu}{n}+a)b_{\nu,n}(y;a,b)1_{[a,b]}(y)
\end{eqnarray}
where for $a \leq y \leq b$, $\nu=1,2,\ldots,n$ the functions
\begin{equation*}
b_{\nu,n}(y;a,b)=b_{\nu,n}\left( \frac{y-a}{b-a}\right)=\frac{(-1)^{n-\nu}}{(b-a)^n} \left(\begin{array}{c}
                                                                               n \\
                                                                               \nu
                                                                             \end{array}
                                                                             \right)
(y-a)^{\nu}(y-b)^{n-\nu}
\end{equation*}
 are the Bernstein polynomials of order $n$ on $[a,b]$ and $b_{\nu,n}(y)=b_{\nu,n}(y;0,1)$\\
Consequently the  Bernstein  approximation of n-th order truncated on $[a,b]$ for the price $C_S$ of a basket option with payoff $h(S_T)$  is defined as:
\begin{equation}\label{berntapprox3}
  \tilde{C}_{B}(a,b,n):= w_1  E_{\mathcal{Q}}\left[e^{-(r-\frac{1}{2}\sigma^2) T+\mu({Y_T^{(2)})}} \tilde{C}_B(Y^{(2)}_T,a,b,n) \right]
\end{equation}
The next theorem provides an expression for the approximation defined above.
\begin{theorem}\label{theo:bernstein}
The Bernstein  approximation  of  n-th order for the price $C_S$ of a spread contract with payoff $h(S_T)$  under the model (\ref{eq:bscmultid}) is given by:
\begin{eqnarray}\notag
  \tilde{C}_{B}(a,b,n)& =& w_1 \left( \frac{(r-\frac{1}{2}\sigma_2^2)T-b}{b-a} \right)^n \\
  &&\sum_{\nu=0}^n  \left(\begin{array}{c}
                                                                               n \\
                                                                               \nu
                                                                             \end{array}\right)
(-1)^{n-\nu} C((b-a)\frac{\nu}{n}+a) \left( \frac{b-a}{(r-\frac{1}{2}\sigma_2^2)T-b} \right)^{\nu}G(\nu)\\ \label{eq:approxgen}
&&
\end{eqnarray}
where:
\begin{eqnarray*}
 G(\nu)  &=& \sum_{k=0}^{\nu}  \left(\begin{array}{c}

                                                                                \nu\\
                                                                                k
                                                                             \end{array} \right)\left( \frac{(r-\frac{1}{2}\sigma_2^2)T-b}{b-a} \right)^k H(n-\nu+k)\\
 H(l)&=&   \sum_{m=0}^{l} \left(\begin{array}{c}
                                 l \\
                                 m
                                 \end{array}\right)
                                                                             \left( \frac{\sqrt{T}\sigma_2}{(r-\frac{1}{2}\sigma_2^2)T-b} \right)^m\left((r-\frac{1}{2}\sigma_2^2)T-b\right)^{n-\nu+k} F(m)\\
 F(m)&=& e^{-\frac{u^2}{2}} \frac{d^m M_Z(u,\tilde{a},\tilde{b})}{d u^m}\mid_{u=\sigma1 \rho \sqrt{T}}=\sum_{k=0}^{m} \left( \begin{array}{c}
                                                                                m\\
                                                                                k
                                                                             \end{array} \right)
(\sigma_1 \rho \sqrt{T})^{m-k}\mu_{\tilde{a}-\sigma_1 \rho \sqrt{T},\tilde{b}-\sigma_1 \rho \sqrt{T}}(k)
\end{eqnarray*}
and
\begin{equation*}\label{eq:bschkvariab}
    C(y):=C_{BS}(K(y),\sigma,S_0^{(1)})
\end{equation*}
 is the Black-Scholes price of a Call Option written on the underlying asset price $S_T^{(1)}$ with strike price $K(y)$, maturity at $T>0$, volatility  $\sigma$, spot price $S_0^{(1)}$ and strike price:
  \begin{equation}\label{}
    K(y)=\frac{1}{w_1} e^{-A} \left( Ke^{-\frac{\sigma_1}{\sigma_2}\rho y}- w_2 S_0^{(2)}e^{(1-\frac{\sigma_1}{\sigma_2}\rho)y} \right)
  \end{equation}
  with:
\begin{eqnarray}\label{constanta}
    A &=& - \frac{1}{2}  \frac{\sigma_1}{\sigma_2} \rho (\rho \sigma_1 \sigma_2 +2r -\sigma_2^2)T\\ \notag
    \sigma &=& \sqrt{1-\rho^2} \sigma_1
   \end{eqnarray}
   and
 \begin{equation*}
    \tilde{a}=\frac{a-(r-\frac{1}{2}\sigma_2^2)T}{\sqrt{T}\sigma_2},\;\;\;\tilde{b}=\frac{ b-(r-\frac{1}{2}\sigma_2^2)T}{\sqrt{T}\sigma_2}
 \end{equation*}

\end{theorem}
\begin{proof}
  By Lemma \ref{lemma1} we replace equation (\ref{eq:bernsexp}) into equation (\ref{berntapprox3}) to have:
  \begin{eqnarray} \nonumber
  && \tilde{C}_{B}(a,b,n) = w_1  E_{\mathcal{Q}}\left[e^{-(r-\frac{1}{2}\sigma^2) T+\mu({Y_T^{(2)})}}\tilde{C}_B(Y_T^{(2)},a,b,n)\right]\\ \nonumber
  &=& w_1 (b-a)^{-n} \sum_{\nu=0}^n  \left(\begin{array}{c}
                                                                                  n\\
                                                                                  \nu
                                                                                \end{array}\right)                                                                                 (-1)^{n-\nu} C_{S}((b-a)\frac{\nu}{n}+a) \\ \nonumber
  && E_{\mathcal{Q}}\left[e^{-(r-\frac{1}{2}\sigma^2)T+\mu(Y_T^{(2)})}\left(Y^{(2)}_T-a \right)^{\nu}\left(Y^{(2)}_T-b \right)^{n-\nu}1_{[a, b]}(Y^{(2)}_T)\right] \label{eq:mixmom1}
  \end{eqnarray}
Now:
\begin{eqnarray} \notag
 && E_{\mathcal{Q}}\left[e^{-rT+\frac{1}{2}\sigma^2T+\mu(Y_T^{(2)})}\left(Y^{(2)}_T-a \right)^{\nu}\left(Y^{(2)}_T-b \right)^{n-\nu}1_{[a, b]}(Y^{(2)}_T)\right] \\ \notag
 &=&e^A \sum_{k=0}^{\nu} \left(\begin{array}{c}
                                                                                  \nu \\
                                                                                  k
                                                                                \end{array}
 \right) (b-a)^{\nu-k} E_{\mathcal{Q}}\left[e^{\frac{\sigma1}{\sigma2}\rho Y_T^{(2)}}\left(Y^{(2)}_T-b\right)^{n-\nu+k}1_{[a, b]}(Y^{(2)}_T)\right] \\
  && \label{eq:expinter}
 \end{eqnarray}
 Moreover we have:
 \begin{eqnarray} \nonumber
 &&  E_{\mathcal{Q}}\left[ e^{\frac{\sigma_1}{\sigma_2}\rho Y_T^{(2)} } \left(Y_T^{(2)}-b \right)^{n-\nu+k} 1_{[a, b]}(Y^{(2)}_T)\right]\\ \nonumber
 &=&    \sum_{m=0}^{n-\nu+k} \left(\begin{array}{c}
                    n-\nu+k \\
                    m
                  \end{array} \right) \left((r-\frac{1}{2}\sigma_2^2 )T-b\right)^{n-\nu+k-m}E_{\mathcal{Q}} \left[ e^{\frac{\sigma_1}{\sigma_2}\rho Y_T^{(2)} }\left(Y_T^{(2)}-E_{\mathcal{Q}}(Y_T^{(2)}) \right)^m 1_{[a, b]}(Y^{(2)}_T)\right]\\ \nonumber
 &=& e^{\frac{\sigma_1}{\sigma_2}\rho (r-\frac{1}{2}\sigma_2^2)T}\left((r-\frac{1}{2}\sigma_2^2)T-b\right)^{n-\nu+k} \sum_{m=0}^{n-\nu+k} \left(\begin{array}{c}
                    n-\nu+k \\
                    m
                  \end{array} \right) \left(\frac{\sqrt{T}\sigma_2}{(r-\frac{1}{2}\sigma_2^2)T-b}\right)^m  \\ \nonumber
                   && E_{\mathcal{Q}} \left[  e^{\sigma_1 \rho \sqrt{T} Z} Z^m 1_{[\tilde{a}, \tilde{b}]}(Z)\right] \\ \label{eq:mgf2}
                   &&
 \end{eqnarray}
 Next, notice that:
 \begin{eqnarray}\nonumber
&& E_{\mathcal{Q}} \left[ e^ {u Z} Z^m 1_{[\tilde{a}, \tilde{b}]}(Z)\right] =  \frac{d^m M_Z\left( u,\tilde{a}, \tilde{b}\right)}{du^m}\\ \nonumber
&=& \frac{1}{\sqrt{2 \pi}}\int_{\tilde{a}}^{\tilde{b}}e^{-\frac{1}{2}(x^2-2 u x)}x^m dx = e^{\frac{u^2 }{2} }\frac{1}{\sqrt{2 \pi}}\int_{\tilde{a}}^{\tilde{b}}e^{-\frac{1}{2}(x- u)^2}x^m dx\\ \nonumber
&=& e^{\frac{u^2 }{2} }\frac{1}{\sqrt{2 \pi}}\int_{\tilde{a}-u}^{\tilde{b}-u}e^{-\frac{1}{2}y^2}(y+u)^m dy\\ \nonumber
&=& e^{\frac{u^2 }{2} } \sum_{\nu=0}^m \left(\begin{array}{c}
                    m \\
                    \nu
                  \end{array} \right)u^{m-\nu}\mu_{\tilde{a}-u, \tilde{b}-u}(\nu)\\ \label{eq:mgftrunc1}
                  &&
  \end{eqnarray}
 Finally substituting equation (\ref{eq:mgftrunc1}), evaluated at $u=\sigma_1 \rho \sqrt{T}$, into equation (\ref{eq:mgf2}), then equation (\ref{eq:mgf2}) into equation (\ref{eq:expinter}) we get equation (\ref{eq:approxgen} ).
\end{proof}
  \begin{remark}
  Notice that the approximation $\tilde{C}_{B}(a,b,n)$ depends only on the values of a Black-Scholes option price on a partition of the interval $[a,b]$  and the truncated mixed exponential-power moments of a Gaussian multivariate distribution.
  \end{remark}
    \section{Approximation by Chebyshev Polynomials}
We study an alternative approximation of the price via Chebyshev polynomials. For definition and their basic properties see for example  Mason and Handscomb (2003).\\
 Denoting by $(T_k(y))_{k \in \mathbb{N}}$ the sequence of Chebyshev polynomials of first type on $[-1,1]$ we consider the n-th approximation of the function $C(y)$  on the interval $[a,b]$ described by equation  in terms of Chebyshev polynomials the one given by:
 \begin{equation}\label{eq:bscheby}
    C_{CH}(y,a,b,n)=\frac{1}{2}\hat{c}_0 1_{[a,b]}(y)+\sum_{k=1}^n \hat{c}_k T_k^{a,b}(y)1_{[a,b]}(y)
 \end{equation}
 where $(T^{a,b}_k (x))_{k \in \mathbb{N}}$ is the sequence of Chebyshev polynomials of first type on $[a,b]$ defined by:
 \begin{equation*}\label{}
    T_k^{a,b}(y)=T_k \left(-1+\frac{2}{(b-a)}(y-a)\right),\;\;\;\; a \leq y \leq b
 \end{equation*}
 and the values $(\hat{c}_k)_{0 \leq k \leq N}$ are estimators of the corresponding coefficients in the Chebyshev expansion.\\
  Chebyshev polynomials on $[a,b]$ are orthogonal with respect to the scalar product defined as:
 \begin{equation*}
    <f,g>=\int_a^b f(x) g(x) w_{a,b}(x) dx
 \end{equation*}
 with weight function $w_{a,b}(x)=\left(1-\left(\frac{2(x-a)}{b-a}-1 \right)^2 \right)^{-\frac{1}{2}}$. \\
Notice that
\begin{equation*}
||T^{a,b}_k ||^2   = \left \{
\begin{array}{cc}
 \frac{(b-a) \pi}{4}  & \text{for}\;\; k \neq 0 \\
 \frac{(b-a) \pi}{2}  &   \text{for}\;\; k = 0
\end{array}
\right.
 \end{equation*}
Then for $k \neq 0$ the coefficients in the expansion can be calculated as:
 \begin{eqnarray*}
   c_k &=&  \frac{<C,T^{a,b}_k>}{||T^{a,b}_k ||^2 }=\frac{4}{(b-a)\pi} \int_a^b C(y)T_k^{a,b}(y) w_{a,b}(y) dy\\
   &=& \frac{4}{(b-a)\pi} \int_a^b C(y)T_k(-1+\frac{2}{(b-a)}(y-a)) w_{a,b}(y) dy\\
   &=& \frac{2}{\pi} \int_{-1}^1 C(a+ \frac{b-a}{2}(x+1))T_k(x) w_{-1,1}(x) dx\\
   &=& \frac{2}{\pi} \int_{0}^{\pi} C(a+ \frac{b-a}{2}(\cos \theta +1))\cos(k \theta) d \theta\\
 \end{eqnarray*}
after changes of variables $x=-1+\frac{2}{(b-a)}(y-a)$ and  $x=\cos \theta$.\\
Also:
\begin{equation*}
    c_0=\frac{1}{\pi} \int_{0}^{\pi} C(a+ \frac{b-a}{2}(\cos \theta +1)) d \theta
\end{equation*}
From  the trapezoidal rule to approximate Riemnan integrals the coefficients $(c_k)_{0 \leq k \leq N}$ can be estimated by an equidistant partition of $N$ points on $[0,\pi]$:
\begin{eqnarray*}\label{}
    \hat{c}_k & \simeq & \frac{b-a}{ N \pi} \sum_{j=0}^N d^{(k)}_j,\;\;\;\;k=1,2,\ldots,n\\
    \hat{c}_0  & \simeq & \frac{b-a}{2 N \pi} \sum_{j=0}^N d^{(0)}_j
\end{eqnarray*}
where:
\begin{equation*}
    d^{(k)}_j= \left\{ \begin{array}{ll}
                    C(b) & \text{for}\;j=0\\
                     2 C(a+ \frac{b-a}{2}(\cos(\frac{\pi j}{N})  +1))\cos (\frac{\pi k j}{N}) & \text{for}\;j=1,2,\ldots,N-1\\
                    C(a)  & \text{for}\; j=N
                  \end{array}
                  \right.
\end{equation*}
Chebyshev polynomials of first type can be written in terms of powers of the variable. From Amparo et al.(2007):
\begin{equation}\label{eq:chebpower}
    T^{a,b}_k(x)=\sum_{l=0}^{[\frac{k}{2}]} b_l^{(k)} \left( -1+\frac{2}{b-a}(x-a) \right)^{k-2l}
\end{equation}
where:
\begin{equation*}
b_l^{(k)}=\left \{ \begin{array}{ll}
                  (-1)^{\frac{k}{2}}   & \text{for}\;l=\frac{k}{2}\;\;\text{and}\;\;k
                  \;\;\text{even} \\
          (-1)^l2^{k-2l-1} \frac{k}{k-l}  \left(\begin{array}{c}
                    k-l \\
                    l
                  \end{array} \right) & \text{for}\; l < \frac{k}{2}
                  \end{array}
                  \right.
  \end{equation*}
In a similar way to the Bernstein polynomial approach we define the n-th order Chebyshev approximation for the spread price as:
\begin{equation*}
  \hat{C}_{CH}(a,b,n):= w_1  E_{\mathcal{Q}}\left[e^{-(r-\frac{1}{2}\sigma^2) T+\mu({Y_T^{(2)})}} \tilde{C}_{CH}(Y^{(2)}_T,a,b,n) \right]
\end{equation*}
The next theorem provides the Chebyshev approximation for the price  of a Basket option:
\begin{theorem}\label{teo:chebappr}
The n-th order Chebyshev's  approximation  of the price $C_S$ of a European basket option under the model given by equation (\ref{eq:bscmultid}) with payoff $h(S_T)$ is given by:
\begin{eqnarray} \nonumber
  \hat{C}_{CH}(a,b,n)& =&  \frac{w_1}{2}\hat{c}_0 [N(\tilde{b}-\sigma_1 \rho \sqrt{T})-N(\tilde{a}-\sigma_1 \rho \sqrt{T})]\\ \nonumber
  &+&  w_1 e^{-\frac{1}{2}\rho^2 \sigma_1^2 T} \sum_{k=1}^n \sum_{l=0}^{[\frac{k}{2}]} \hat{c}_k b_l^{(k)} \left(\frac{2\sigma_2 \sqrt{T}}{b-a} \right)^{k-2l} \hat{G}(k-2l)\\ \label{eq:chebprice}
  &&
\end{eqnarray}
where:
\begin{eqnarray*}
\hat{ G}(k)  &=&  \sum_{m=0}^{k}  \left(\begin{array}{c}
                    k \\
                    m
                  \end{array} \right) \left(\frac{2(r-\frac{1}{2}\sigma_2^2 )T-a-b}{2\sigma_2 \sqrt{T}} \right)^m \frac{d^{k-m} M_{Z}(u,\tilde{a},\tilde{b})}{du^{k-m}}\mid_{u=\sigma_1 \rho \sqrt{T}}
\end{eqnarray*}
for $k=0,1,2,\ldots,n$,
\begin{equation*}
    \tilde{b}=\frac{b-(r-\frac{1}{2}\sigma_2^2 )T}{\sigma_2 T},\;\;\tilde{a}=\frac{a-(r-\frac{1}{2}\sigma_2^2 )T}{\sigma_2 T}
\end{equation*}

 \end{theorem}
  \begin{proof}
  Notice that:
  \begin{eqnarray*}
   e^A E_{\mathcal{Q}}\left[e^{\frac{\sigma1}{\sigma2}\rho Y_T^{(2)}}1_{(-\infty,b)}(Y_T^{(2)}) \right]  &=& e^A  e^{(r-\frac{1}{2}\sigma_2^2 ) \frac{\sigma_1}{\sigma_2}\rho T}E_{\mathcal{Q}}\left[e^{\sigma_1 \rho \sqrt{T}Z}1_{(-\infty,\tilde{b})}(Z) \right]\\
   &=& e^A  e^{(r-\frac{1}{2}\sigma_2^2 ) \frac{\sigma_1}{\sigma_2}\rho T}e^{\frac{1}{2} \sigma_1^2 \rho^2 T}\frac{1}{\sqrt{2 \pi}} \int_{-\infty}^{\tilde{b}-\sigma_1 \rho \sqrt{T}}e^{-\frac{1}{2}z^2}\;dz\\
   &=& N(\tilde{b}-\sigma_1 \rho \sqrt{T})
    \end{eqnarray*}
    Hence:
    \begin{equation*}
     e^A E_{\mathcal{Q}}\left[e^{\frac{\sigma1}{\sigma2}\rho Y_T^{(2)}}1_{[a,b]}(Y_T^{(2)}) \right]= N(\tilde{b}-\sigma_1 \rho \sqrt{T})-N(\tilde{a}-\sigma_1 \rho \sqrt{T})
    \end{equation*}
 From Lemma \ref{lemma1}, we take into account  equations (\ref{eq:bscheby}) and (\ref{eq:chebpower}) to have:
  \begin{eqnarray} \nonumber
   \hat{C}_{CH}(a,b,n)&=& w_1 e^A E_{\mathcal{Q}} \left(e^{\frac{\sigma_1}{\sigma_2}\rho Y^{(2)}_T}\hat{C}(Y^{(2)}_T,a,b,n) \right)\\ \nonumber
    &=&  \frac{w_1}{2}  \hat{c}_0  [N(\tilde{b}-\sigma_1 \rho \sqrt{T})-N(\tilde{a}-\sigma_1 \rho \sqrt{T})] \\ \nonumber
      &+& w_1 e^A  \sum_{k=1}^n \hat{c}_k  E_{\mathcal{Q}}\left[ e^{\frac{\sigma1}{\sigma2}\rho Y_T^{(2)}}T_k^{a,b}(Y^{(2)}_T) 1_{[a,b]}(Y^{(2)}_T) \right] \\ \nonumber
      &=&  \frac{w_1}{2}\hat{c}_0 [N(\tilde{b}-\sigma_1 \rho \sqrt{T})-N(\tilde{a}-\sigma_1 \rho \sqrt{T})] \\ \nonumber
       &+& w_1 e^A  \sum_{k=1}^n \sum_{l=0}^{[\frac{k}{2}]} \frac{\hat{c}_k b_l^{(k)}}{(b-a)^{k-2l}} E_{\mathcal{Q}}\left[e^{\frac{\sigma1}{\sigma2}\rho Y_T^{(2)}}\left( 2Y^{(2)}_T-a-b \right)^{k-2l} 1_{[a,b]}(Y^{(2)}_T)\right]\\ \label{eq:bss1}
       &&
      \end{eqnarray}
     Now:
     \begin{eqnarray*}
    &&  E_{\mathcal{Q}}\left[e^{\frac{\sigma1}{\sigma2}\rho Y_T^{(2)}}( 2Y^{(2)}_T-a-b)^{k-2l} 1_{[a,b]}(Y^{(2)}_T)\right]  \\
    &=&    \sum_{m=0}^{k-2l}\left(\begin{array}{c}
                    k-2l \\
                    m
                  \end{array} \right) (2 E_{\mathcal{Q}}Y^{(2)}_T-a-b)^m 2^{k-2l-m}E_{\mathcal{Q}}\left[e^{\frac{\sigma1}{\sigma2}\rho Y_T^{(2)}}(Y^{(2)}_T-E_{\mathcal{Q}}Y^{(2)}_T)^{k-2l-m} 1_{[a,b]}(Y^{(2)}_T)\right]\\
    &=&  e^{\frac{\sigma_1}{\sigma_2}\rho (r-\frac{1}{2}\sigma_2^2 )T} (2 \sigma_2 \sqrt{T})^{k-2l}\sum_{m=0}^{k-2l}\left(\begin{array}{c}
                    k-2l \\
                    m
                  \end{array} \right) \left(\frac{2(r-\frac{1}{2}\sigma_2^2)T-a-b}{2\sigma_2 \sqrt{T}}\right)^m    \\
    && E_{\mathcal{Q}} \left[ e^{\sigma_1 \rho  \sqrt{T} Z}Z^{k-2l-m} 1_{[\tilde{a}, \tilde{b}]}(Z)\right]\\ \nonumber
  &=&  e^{\frac{\sigma_1}{\sigma_2}\rho (r-\frac{1}{2}\sigma_2^2 )T} (2 \sigma_2 \sqrt{T})^{k-2l}\sum_{m=0}^{k-2l}\left(\begin{array}{c}
                    k-2l \\
                    m
                  \end{array} \right) \left(\frac{2(r-\frac{1}{2}\sigma_2^2)T-a-b}{2\sigma_2 \sqrt{T}}\right)^m    \\
    && \frac{d^{k-2l-m} M_Z(u,\tilde{a},\tilde{b})}{du^{k-2l-m}} \mid_{u=\sigma_1 \rho \sqrt{T}}
    \end{eqnarray*}
     Substituting the last expression into equation (\ref{eq:bss1}) we get equation (\ref{eq:chebprice})
    \end{proof}
\subsection{Sensitivities with respect to the parameters}
Sensitivities with respect to the parameters in the model and  contract can be analyzed with the help of the approximations above. It is worth noticing that the lack of a closed-form formula of the spread price makes impossible to directly differentiate it to obtain the corresponding \textit{Greek}. On the other hand a Monte Carlo approach to compute derivatives of the price with a reasonable error, in practice, requires a considerable  extra amount of computational effort.  Here we focus on the deltas of the spread contract, i.e. the sensitivities with respect to the spot prices given by a Chebyshev approximation. Other sensitivities follows in a similar way.\\
To this end we slightly change the notations,  explicitly stating the dependence on the initial prices $S_0^{(1)}=s_1$ and $S_0^{(2)}=s_2$. Thus, we write $C(y,s_1,s_2)$ instead of $C(y)$ in equation (\ref{eq:condprice3}) and $K(y,s_2)$ for $K(y)$.\\
The conditional price can be written now:
\begin{equation*}\label{}
    C(y,s_1,s_2)= s_1 N(d_1(s_1, K(y,s_2))-K(y,s_2)e^{-rT} N(d_2(s_1, K(y,s_2))
\end{equation*}
where:
\begin{eqnarray*}\label{}
    d_1(s_1,K(y,s_2))&=&\frac{\ln \left(\frac{s_1}{K(y,s_2)}\right)+(r+\frac{\sigma^2}{2})T}{\sigma \sqrt{T}}\\
    d_2(s_1,K(y,s_2))&=& d_1(s_1,K(y,s_2))-\sigma \sqrt{T}
\end{eqnarray*}
Then by elementary differentiation:
\begin{eqnarray*}
   \frac{\partial C}{\partial s_1}&=& s_1 f_Z( d_1(s_1,K(y,s_2)))\frac{\partial d_1}{\partial s_1}-e^{-rT}K(y,s_2)f_Z( d_2(s_1,K(y,s_2)))\frac{\partial d_2}{\partial s_1}+ N(d_1(s_1,K(y,s_2)))\\
   &=& \frac{1}{\sigma \sqrt{T}} \left[f_Z( d_1(s_1,K(y,s_2)))-e^{-rT}\frac{K(y,s_2)}{ s_1}f_Z( d_2(s_1,K(y,s_2))) \right]+N(d_1(s_1,K(y,s_2)))\\
   \frac{\partial C}{\partial s_2}&=& s_1 f_Z( d_1(s_1,K(y,s_2)))\frac{\partial d_1}{\partial s_2}\\
   &-&e^{-rT}K(y,s_2)f_Z( d_2(s_1,K(y,s_2)))\frac{\partial d_2}{\partial s_2}+ e^{-rT}\frac{ \partial K(y,s_2)}{\partial s_2} N( d_2(s_1,K(y,s_2)))\\
   &=&\frac{w_2}{w_1}e^{-A}e^{(1-\frac{\sigma_1}{\sigma_2}\rho)y} \left[\frac{1}{\sigma \sqrt{T} w_1} \left(-\frac{s_1}{K(y,s_2)}f_Z( d_1(s_1,K(y,s_2)))+e^{-rT} f_Z( d_2(s_1,K(y,s_2))) \right) \right.\\
   &-& \left. e^{-rT}N( d_2(s_1,K(y,s_2)))\right]
\end{eqnarray*}
where  $f_Z$ and $N$ are respectively the density and cumulative distribution functions of a standard normal random variable.\\
Taking into account equation (\ref{eq:bscheby}) a natural estimator of $\frac{\partial C(y)}{\partial s_j}$ is:
\begin{equation*}\label{eq:bschebyder}
  \frac{\partial C_{CH}(y,a,b,n)}{\partial s_j}  =\frac{1}{2}\frac{\partial \hat{c}_0 }{\partial s_j}+\sum_{k=1}^n  \frac{\partial \hat{c}_k }{\partial s_j}T_k^{a,b}(y)\;\;\;\text{for}\;\;j=1,2.
 \end{equation*}
Notice that $C(y,s_1,s_2)$ is continuously differentiable with respect to \\
$s_1$ and $s_2$, then derivative and integral can be interchanged. Therefore we have:
\begin{equation*}\label{}
  \frac{\partial \hat{c}_k }{\partial s_j}= \frac{2}{\pi} \int_{0}^{\pi} \frac{\partial C }{\partial s_j}\left(a+ \frac{b-a}{2}(\cos \theta +1),s_1,s_2 \right)\cos(k \theta) d \theta
\end{equation*}
for $k=1,2,\ldots,n$.\\
The coefficients in the development are estimated by:
\begin{eqnarray*}\label{}
    \frac{\partial \hat{c}_k }{\partial s_j} & \simeq & \frac{b-a}{ N \pi} \sum_{j=0}^N \frac{\partial d^{(k)}_j(s_1,s_2) }{\partial s_j}\;\;\;\;k=1,2,\ldots,n\\
    \frac{\partial \hat{c}_0 }{\partial s_j}  & \simeq & \frac{b-a}{2 N \pi} \sum_{j=0}^N \frac{\partial d^{(0)}_j(s_1,s_2) }{\partial s_j}
\end{eqnarray*}
where:
\begin{equation*}
  \frac{\partial   d^{(k)}_l }{\partial s_j}(s_1,s_2)= \left\{ \begin{array}{ll}
    2 \frac{\partial C}{\partial s_j}(a+ \frac{b-a}{2}(\cos(\frac{\pi l}{N})  +1))\cos (\frac{\pi k l}{N}) & \text{for}\;l=1,2,\ldots,N-1\\
                    0  & \text{for}\; l=0,N
                  \end{array}
                  \right.
\end{equation*}
Finally we estimate $\Delta^{(j)} C_S$  by:
\begin{eqnarray} \notag
  \Delta^{(j)} C_S&=&\frac{\partial  \hat{C}_{CH}(a,b,n) }{\partial s_j} =  \frac{w_1}{2}\frac{\partial \hat{c}_0 }{\partial s_j}[N(\tilde{b}-\sigma_1 \rho \sqrt{T})-N(\tilde{a}-\sigma_1 \rho \sqrt{T})]\\ \notag
   &+&  w_1 e^{-\frac{1}{2}\rho^2 \sigma_1^2 T} \sum_{k=1}^n \sum_{l=0}^{[\frac{k}{2}]} \frac{\partial \hat{c}_k }{\partial s_j}  b_l^{(k)} \left(\frac{2\sigma_2 \sqrt{T}}{(b-a)} \right)^{k-2l} \hat{G}(k-2l)\\ \label{eq:chebdelta}
  &&
\end{eqnarray}
\begin{figure}[htb!]
  \includegraphics[width=12 cm, height=10 cm]{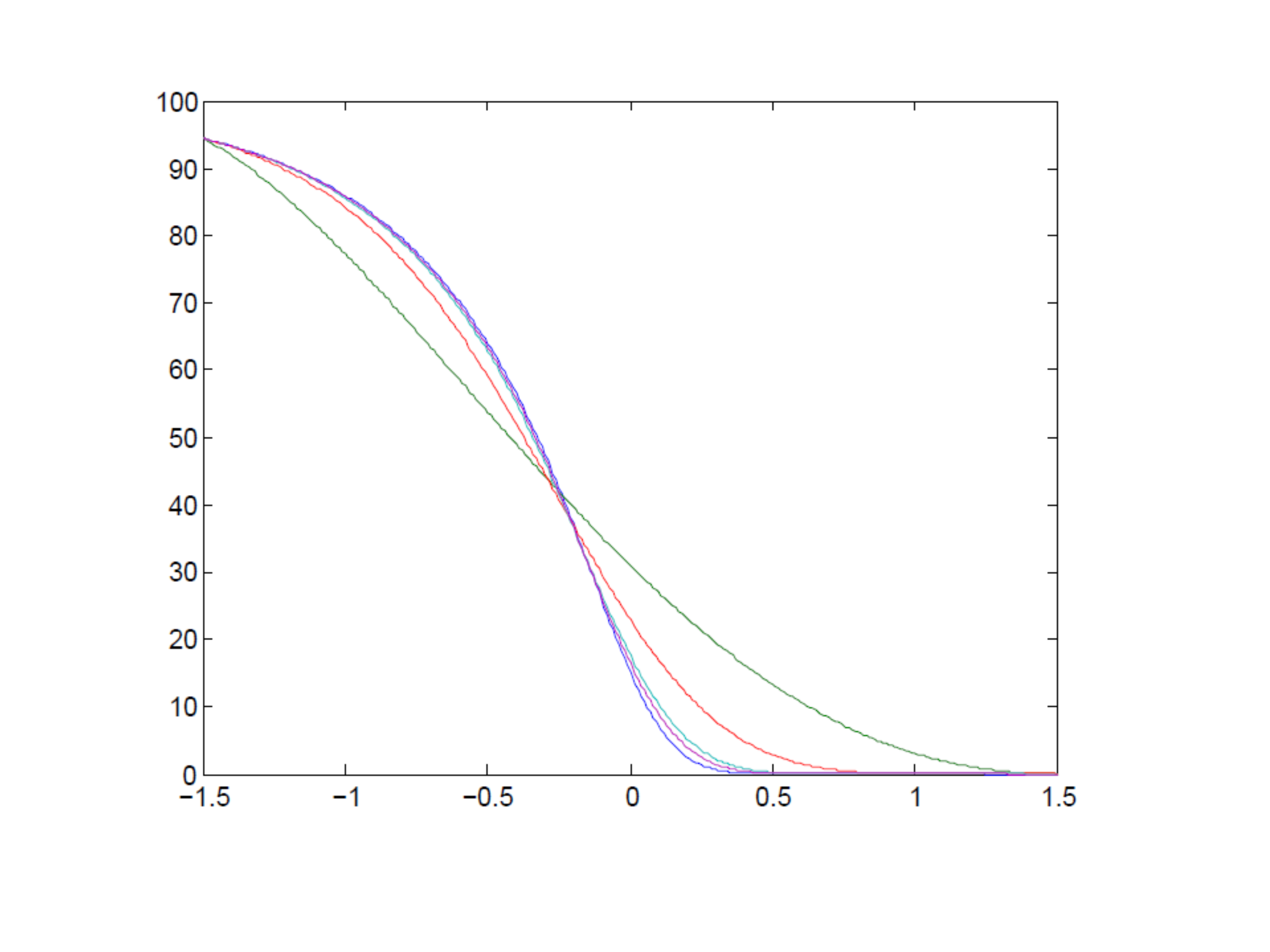}
  \caption{Approximations of the conditional price based on Bernstein polynomials using the benchmark set of parameters. Blue, green, red, light blue and magenta represents the conditional price, and approximations for $n=4$, $n=10$, $n=100$ and $n=200$ respectively. The truncation interval is $[-1.5,1.5]$. }\label{figbernsteinapprox}
\end{figure}
\begin{figure}[hb!]
  \includegraphics[width=12 cm, height=10 cm]{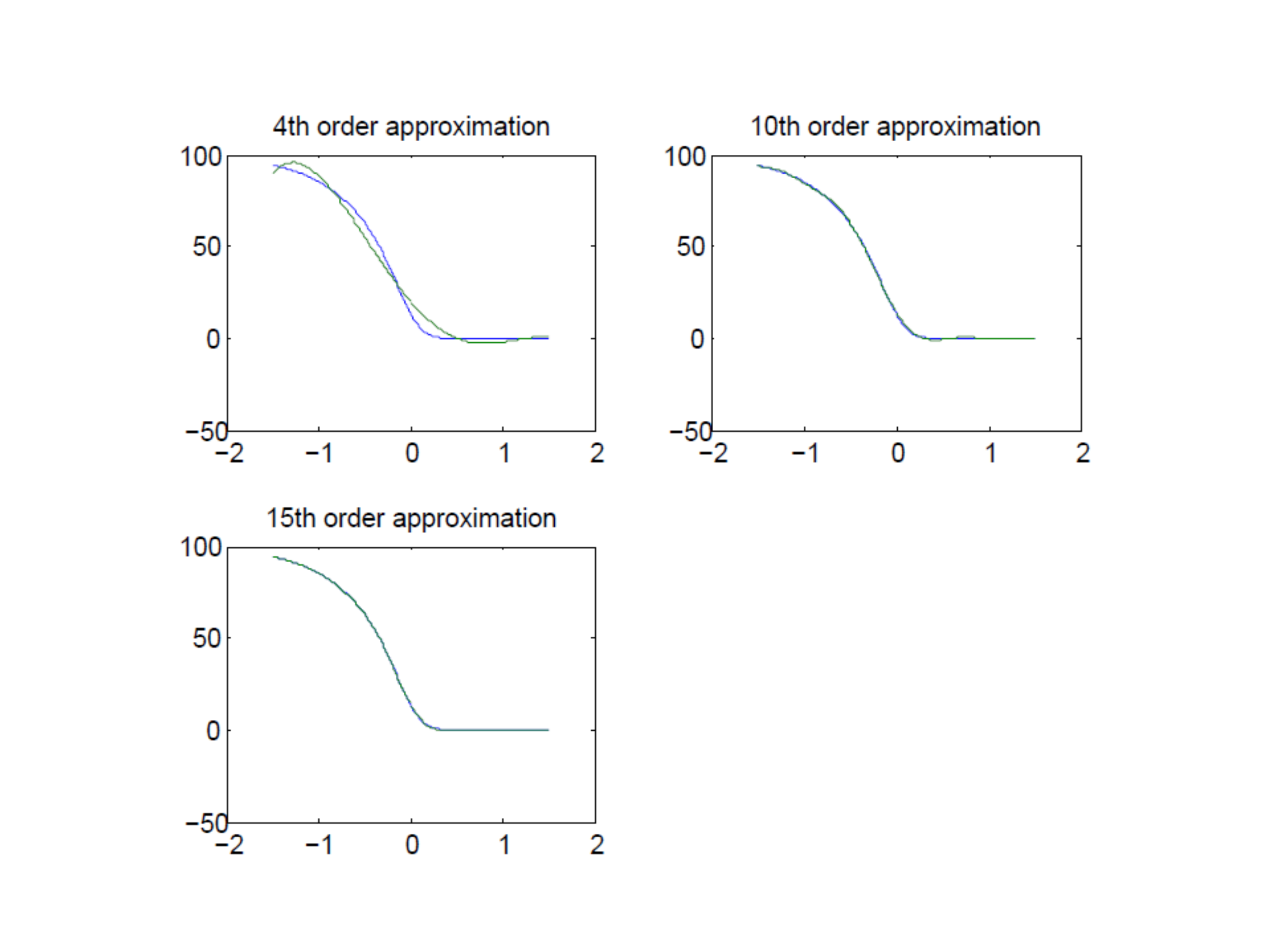}
  \caption{Approximations of the conditional price based on Chebyshev polynomials using the benchmark set of parameters. Clockwise, from top left are 4th, 10th and 15th order respectively.  The truncation interval is $[-1.5,1.5]$. }\label{figchebyapprox}
\end{figure}
\begin{figure}[htb!]
  \includegraphics[width=12 cm, height=10 cm]{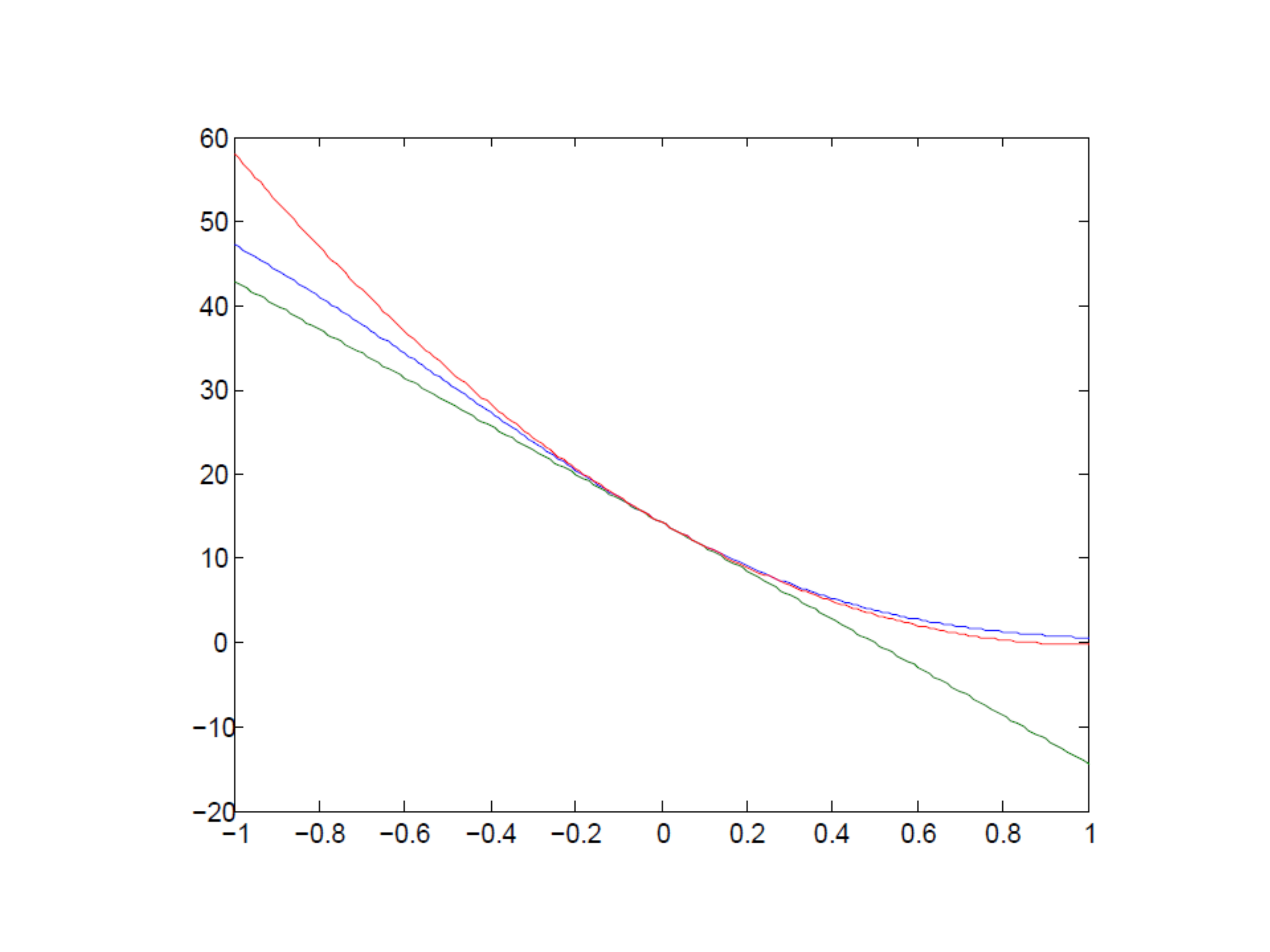}
  \caption{Approximations of the conditional price based on Taylor polynomials using the benchmark set of parameters. The blue line describes the conditional price on $[-1,1]$ while the green and red lines show first and second order expansions respectively. }\label{figtaylorapprox}
\end{figure}
\begin{figure}[htb!]
  \includegraphics[width=12 cm, height=10 cm]{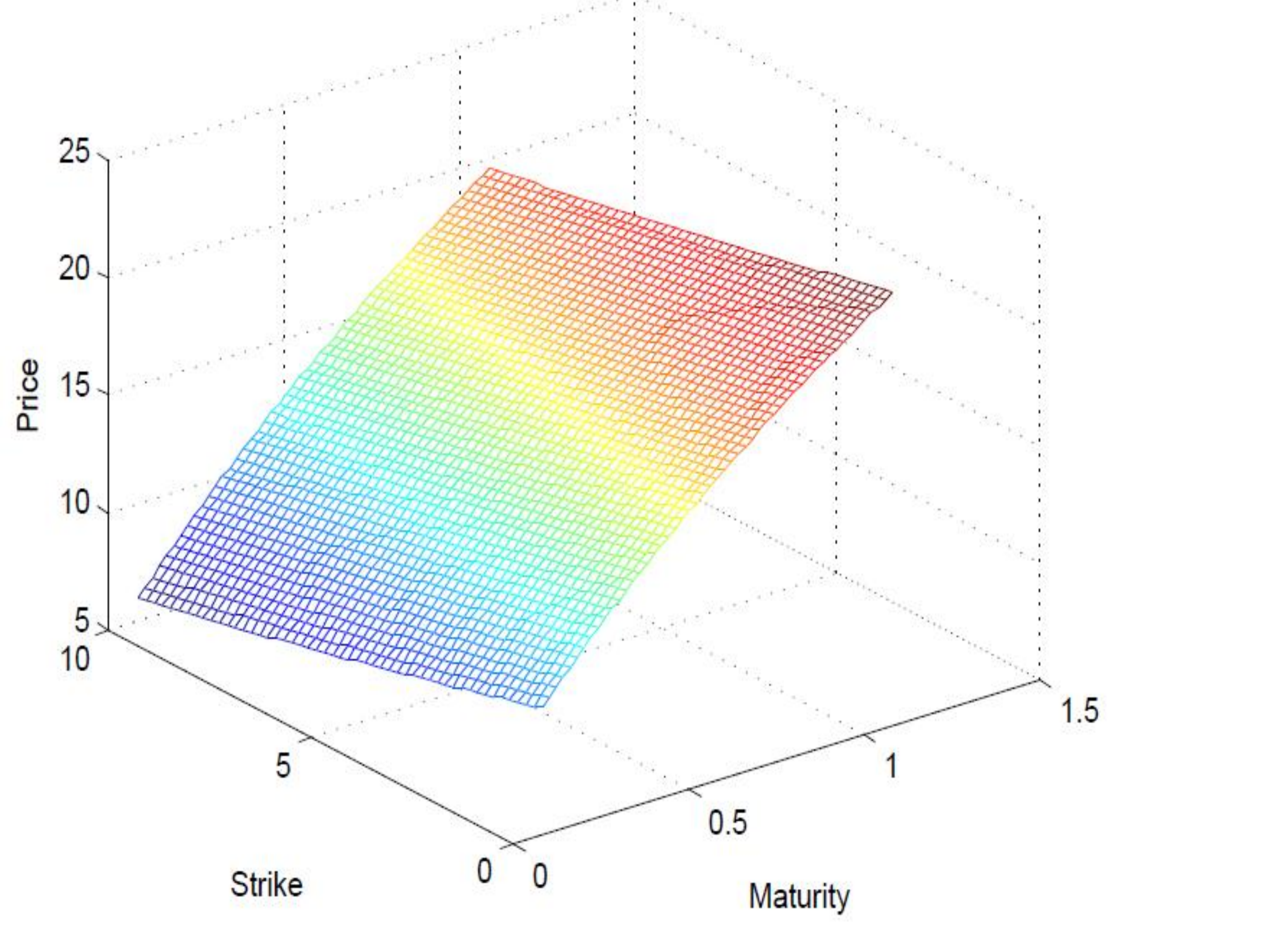}
  \caption{Prices of spreads with maturity between one month and one year and strike price between 0 and 10 dollars under the benchmark parametric set. }\label{figpricevstk}
\end{figure}
 \begin{figure}[htb!]
  \includegraphics[width=12 cm, height=10 cm]{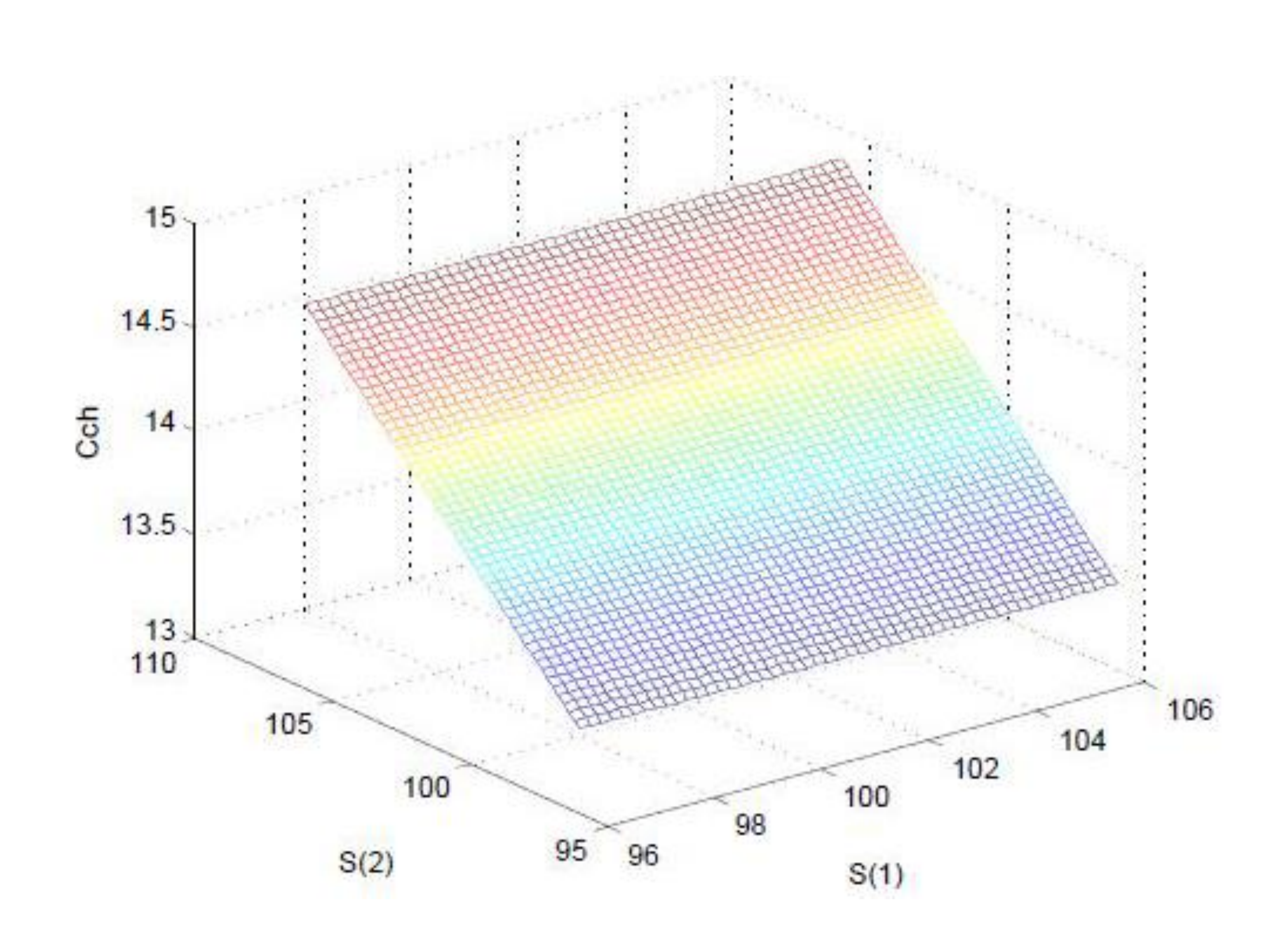}
  \caption{ Prices of spreads under the benchmark parametric set and initial prices ranging from 96-106 dollars. }\label{figpricevsinipri}
\end{figure}
 \begin{figure}[htb!]
  \includegraphics[width=12 cm, height=10 cm]{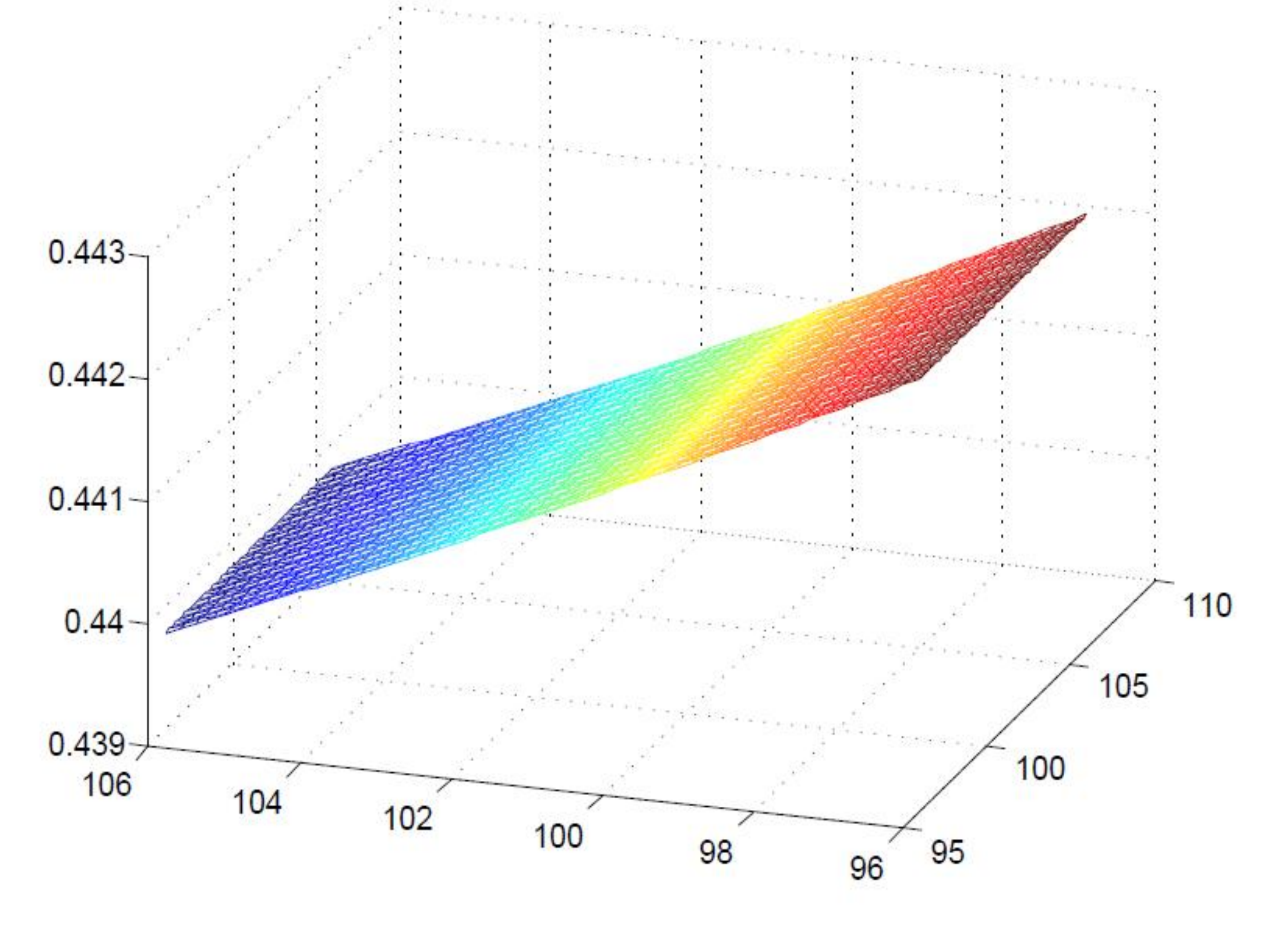}
  \caption{ Delta of the spread with respect to $s_2$  under the benchmark parametric set and initial prices ranging from 96-106 dollars. }\label{figdeltas2}
\end{figure}
\section{Numerical results}
As benchmark setting we consider  initial prices of $S^{(1)}_0=100$ and $S^{(2)}_0=96$ dollars, strike price of $K=\$1$, maturity at $T=1$ year and an annual interest rate of $r=3\%$ under a bivariate Black-Scholes model, with a negative correlation $\rho=-0.3$  and respective volatilities equal to $\sigma_1=0.3$ and $\sigma_2=0.1$. In Figure \ref{figbernsteinapprox} different approximations of the conditional price $C(y)$ based on Bernstein polynomials are shown. The blue line represents the actual conditional price, the green line describes the corresponding expansion  of order $n=4$, while the red line is the approximation for $n=10$.\\
 In order to achieve an accurate approximation polynomials with  high orders, e.g. $n=100$ or even $n=200$, are required. This limits the applicability of the method. Moreover, numerical instabilities appear in the computation of large factorials, even when an asymptotic Stirling formula is used in their computation.\\
 A more promising result is obtained when an approximation of the conditional price is done via Chebyshev polynomials. Figure \ref{figchebyapprox} represents, clockwise from top left, approximations of 4th, 10th and 15th order respectively. Expansions of order 10th and 15th are practically indistinguishable from the original function. Coefficients in the expansion are calculated following a trapezoidal rule with $100$ points on the interval $[-1.5,1.5]$.
For  comparison in Figure \ref{figtaylorapprox} we show  the conditional price function on $[-1,1]$ for the benchmark setting (blue line), together with the  first and second Taylor expansions around the mean value of $Y^{(2)}_T$, i.e. $E_{\mathcal{Q}}=Y^{(2)}_T= (r-\frac{1}{2} \sigma_2^2)T$ (lines green and red respectively). While the second order approximation  offers a  reasonable local fit, significant errors may be found for values far from the mean, which in turns may impact in the price given by its conditional expected value under the risk neutral probability when, for example volatilities are high. See Alvarez and Olivares (2014) for details.\\
In Table \ref{tab1} we show prices of spread contracts obtained under the benchmark parameter set,  the correlation between the two assets varies. Prices are computed by a Monte Carlo approach based on 10 millions simulations of the asset prices with correlated Brownian motions (column 2). In addition we implement a second order Taylor expansion and an approximation by mean of Chebyshev polynomials of order $n=15$. We consider positive and negative correlations, large moderated and weak correlations. In all cases the Chebyshev approximation shows a notable agreement with Monte Carlo prices at less computational cost. \\

\begin{table}
  \centering
  \begin{tabular}{|c|c|c|c|c|}
    \hline
Correlation & Monte Carlo &   Taylor Second approx.& Chebyshev, $n=15$\\ \hline
$\rho=-0.1$ & 14.292128&  13.87090 &14.29060779 \\ \hline
$\rho =0.1$&13.56278&14.78882&13.5649 \\ \hline
 $\rho =-0.3$   &  14.9734 & 15.0065 &14.96293\\
\hline
   $\rho =0.3$ & 12.8085& 12.7901 & 12.790289\\
 \hline
   $\rho =-0.5$    &15.6273 & 15.9238  & 15.63157 \\ \hline
 $\rho =0.5$    & 11.9525 &11.9646 &11.9566 \\
 \hline
   $\rho =-0.7$&16.2421 & 17.5217 & 16.25209\\ \hline
      $\rho =0.7$& 11.03146  & 11.194724 &  11.05286\\
    \hline
  \end{tabular}
  \caption{Spread prices for the benchmark parameters and several values of $\rho$, using Monte Carlo, a Taylor second order approximation around $y^*=0$ and a Chebyshev approximation with 15 terms, $a=-4$ and $b=0.25$. }\label{tab1}
\end{table}

 In the range of parameters considered pricing Chebyshev approximation works about seven times faster when compared with a standard Monte Carlo approach. Taylor has a even a lesser computational time, but for large correlations it is not as accurate as the former.\\
 Due to the steepness of the function $C(y)$ the Chebyshev approximation is sensible to the truncation interval $[a,b]$. In our numerical computations we have used $a=-4$ and $b=0.25$. Within the range of parameter considered most values of $Y_T^{(2)}$ lie on the selected interval $[a,b]$, hence truncation does not affect the mixed exponential-power function by much. Otherwise outside the interval $[a,b]$ the truncation to zero can be replaced by $C(a)$ if $y<a$ or $C(b)$ if $y>a$, as the conditional price remains flat outside a convenient interval. \\
  The method is  stable for the number of points considered in the trapezoid rule. Also approximation gets close to the actual price after a fairly moderate  number of polynomials. For $n=10$ the method shows a good approximation within an error in the order of a penny. For $n=15$ and $n=20$ the approximation improves even more. For approximations of larger orders the gain in precision does not compensate the increase in computational time.\\
   Figure  \ref{figpricevstk} shows prices of a spread contract based on a Chebyshev approximation of order $n=10$. Maturity times ranges from one month to one year, while strike prices go from zero (exchange option) to 10 dollars. Results are consistent with a linear increase in the contract prices with higher maturity and  their decrease with the increase of the strike price.\\
  Figure \ref{figpricevsinipri} shows the price as function of both prices while Figure  \ref{figdeltas2} provides the delta of the spread contract with respect to the second underlying and a range of values of initial prices going from 96 to 106 dollars. Delta values are calculated according to equation (\ref{eq:chebdelta}).
  \newpage
\section{Conclusions}
We compare three methods to price spreads options under a bivariate Black-Scholes model with correlated Brownian motions versus a standard Monte Carlo approach. Our results show that Bernstein approximation requires a large number of terms to achieve a good precision, while Taylor approximation does not offer a uniform convergence, hence a poor result when values are far from the point around the expansion is taken. For some values of the parameter set it may affect the corresponding expected value.\\
 The approximation based on Chebyshev polynomials seems to be appropriate in terms of the  balance offered between accuracy and computational cost. Moreover, the method is suitable to be implemented in more general models provided the conditional distribution is available.

\end{document}